\documentclass[11pt]{article}
\usepackage{amsmath}
\usepackage{amsfonts,times}
\usepackage{cleveref}
 \usepackage{times} 
 \usepackage{epsfig}
\usepackage[mathscr]{euscript}
\usepackage{minibox}

\usepackage{url}
\usepackage{color,microtype}
\usepackage{graphicx,wrapfig,subfigure,titling}
\usepackage{amssymb,amsfonts}
\usepackage{amsthm}
\usepackage{textcomp}
\usepackage{multirow}
\usepackage{gensymb}
\usepackage{array}
\usepackage{times}
\usepackage[sort]{cite}

\usepackage{fullpage}
\usepackage{amsmath,amssymb}
\usepackage{enumitem}
\usepackage{float}
\setlist{noitemsep}

\DeclareMathOperator{\scap}{Cap}

\DeclareMathOperator{\ind}{Ind}
\DeclareMathOperator{\fcc}{CP}

\DeclareMathOperator{\vc}{VC}

\DeclareMathOperator{\mm}{MM}
\DeclareMathOperator{\fm}{FM}

\DeclareMathOperator{\is}{\alpha}
\DeclareMathOperator{\cl}{cl}
\DeclareMathOperator{\bo}{bo}

\newtheorem{theorem}{Theorem}

\newtheorem{proposition}[theorem]{Proposition}
\newtheorem{corollary}[theorem]{Corollary}
\newtheorem{lemma}[theorem]{Lemma}

\newcommand{\calC}{{\mathcal C}}

\newcommand{\eat}[1]{}

\newcommand\ff{{\mathbb F}}

\newcommand\integers{{\mathbb Z}}

\date{}

\begin{document}
\title{
Storage Capacity as an Information-Theoretic Vertex Cover \\and the Index Coding Rate
}
\author{
Arya Mazumdar 
\and 
Andrew McGregor
 \and 
Sofya Vorotnikova }



\maketitle
{\renewcommand{\thefootnote}{}\footnotetext{

 

College of Information and Computer Sciences, University of Massachusetts, Amherst. 
\texttt{\{arya,mcgregor,svorotni\}@cs.umass.edu}. This work was supported by NSF  Awards CCF-0953754, IIS-1251110,  CCF-1320719, CCF-1642658, CCF-BSF-1618512, and a Google Research Award. Some of the results of this paper have appeared in the proceedings of the IEEE International Symposium on Information Theory, 2017.
}
\renewcommand{\thefootnote}{\arabic{footnote}}

\begin{abstract}
Motivated by applications in distributed storage, the storage capacity of a graph was recently defined to be the maximum amount of information that can be stored across the vertices of a graph such that the information at any vertex can be recovered from the information stored at the neighboring vertices. Computing the storage capacity is a fundamental problem in network coding and is related, or equivalent, to some well-studied problems such as index coding with side information and generalized guessing games. In this paper, we consider storage capacity as a natural information-theoretic analogue of the minimum vertex cover of a graph. Indeed, while it was known that storage capacity is upper bounded by minimum vertex cover, we show that by treating it as such we can get a $3/2$ approximation for planar graphs, and a $4/3$ approximation for triangle-free planar graphs. Since the storage capacity is intimately related to the index coding rate, we get a $2$ approximation of index coding rate for planar graphs and $3/2$ approximation for triangle-free planar graphs. Previously, only a trivial $4$ approximation of the index coding rate was known for planar graphs. We also show a polynomial time approximation scheme for the index coding rate when the alphabet size is constant. We then develop a general method of ``gadget covering'' to upper bound the storage capacity in terms of the average of a set of vertex covers. This method is intuitive and leads to the exact characterization of storage capacity for various families of graphs. As an illustrative example, we use this approach to derive the \emph{exact} storage capacity of cycles-with-chords, a family of graphs related to outerplanar graphs. Finally, we generalize the storage capacity notion to include recovery from partial node failures in distributed storage. We show  tight upper and lower bounds on this partial recovery capacity that scales nicely with the fraction of failures in a vertex.

\end{abstract}


\section{Introduction}
The Shannon capacity of a graph  \cite{shannon1956zero} is a well studied parameter that quantifies the zero-error
capacity of a noisy communication channel. 
There are also several other notions of graph capacities or graph entropies
that model different communication/compression scenarios (for example, see \cite{alon1996source}).
In this paper, we are interested in a  recent definition of graph capacity, called the {\em storage capacity}, that we consider to be a natural information-theoretic analogue of the minimum vertex cover of a graph.


Suppose, every vertex of a graph can store a symbol (from any alphabet) with the
criterion that the content of any vertex can be uniquely recovered  from the contents of its neighborhood in the graph.
Then the maximum amount of information
that can be stored in the graph is called the {storage capacity} of that graph \cite{mazumdar2015storage}.
This formulation is mainly motivated by applications in distributed storage, and generalizes the popular definition of {\em locally repairable codes} \cite{gopalan2012locality,papailiopoulos2014locally,cadambe2015bounds,ShanmugamD14}.
In a distributed storage system, each symbol (or coordinate) of a codeword vector is stored at a different server or storage node. In the case of a single server
failure, it is desirable to be able to recover the data of that server by accessing a small number of other servers. Given the topology of the storage network as a graph,
it is quite  natural to model the local repair problem as a {\em neighborhood repair} problem as above. 

Formally, suppose we are given an $n$-vertex {\em undirected} graph $G(V,E)$, where $V=[n] \equiv \{1,2,\ldots ,n\}$.
Also,
given a positive integer $q \ge 2$,
 let $H(X)$ be the Shannon entropy of the random variable $X$ in $q$-ary units (for example, when $q=2$, the entropy is in bits). Let   $\{X_i\}_{i\in V}$, be random variables each with a finite sample space $Q$ of size $q$. For any $I \subseteq [n]$, let $X_I \equiv \{X_i: i \in I\}$.  Consider the solution of the following optimization problem: 
\begin{align}
\max H(X_1, \dots , X_n)
\end{align}
such that $$
H(X_i|X_{N(i)})=0,
$$
for all $i\in V$ where $N(i)=\{j\in V: (i,j)\in E\}$ is the set of neighbors of vertex $i$. This is the storage capacity of the graph $G$ and we denote it  by
$\scap_q(G)$. Note that, although we hide the unit of entropy in the notation $H(\cdot)$, the unit should be clear
from context, and the storage capacity should depend on it, as reflected in the subscript in the notation $\scap_q(G)$.
The absolute storage capacity is defined to be,
\begin{align}
\scap(G) \equiv \sup_q \scap_q(G). 
\end{align}
Note that, $
H(X_i|X_{N(i)})=0
$  in the above definition implies that there exist $n$ deterministic functions $f_i: Q^{|N(i)|} \to Q$ such that $f_i(X_{N(i)}) = X_i$, $i = 1, 2, \dots, n.$ These functions are called the {\em recovery functions}. Given the recovery functions we can define a {\em storage code} $\{(x_1, x_2, \dots, x_n) \in Q^n : x_i = f_i(x_{N(i)}) \}.$ It follows that $\scap_q(G)$ is the logarithm of the maximum possible size of a storage code over all possible sets of recovery functions. 

In \cite{mazumdar2015storage}, it was  observed that the storage capacity is  
upper bounded by the size of the minimum vertex cover $\vc(G)$ of the graph $G$.
\begin{align}\label{eq:vc}
\scap(G) \le |\vc(G)|.
\end{align}
The proof of this fact is quite simple. Since all the neighbors of $V\setminus \vc$ belong to $\vc$, 
$$
H(X_V) = H(X_{\vc(G)}, X_{V \setminus \vc(G)}) = H(X_{\vc(G)}) + H( X_{V \setminus \vc(G)}| X_{\vc(G)}) = H(X_{\vc(G)}) \le |\vc(G)|.
$$ 
Indeed, this proof shows that $H(X_V) = H(X_{\vc(G)})$. 
Because of this, we think it is natural to view storage capacity as an  information theoretic analogue of vertex cover. 
It was also shown in \cite{mazumdar2015storage} that the storage capacity is at least equal to the
 size $\mm(G)$ of the maximum matching   of the graph $G$:
\begin{align}
\mm(G) \le \scap_2(G) \le \scap(G). 
\end{align}
 Since maximum matching and minimum vertex cover are  two quantities
 within a factor of two of each other and maximum matching can be found in polynomial time, this fact gives a $2$-approximation
of the storage capacity\footnote{Indeed, finding a {\em maximal matching} is sufficient for this purpose.}. Provable strict improvement of the maximum matching scheme 
 is unlikely to be achieved by simple means, since that would imply a better-than-2 approximation ratio for the minimum vertex cover problem
violating the unique games conjecture \cite{khot2008vertex}. 

This motivates us to look for natural families of graphs where minimum vertex cover has a better approximation. For example, for bipartite graphs
maximum matching is equal to minimum vertex cover and hence storage capacity is exactly equal to the minimum vertex cover. Another obvious class, and our focus in Section~\ref{sec:planar}, is the family of planar graphs for which a PTAS (polynomial-time-approximation-scheme) is known  \cite{Bar-YehudaE82,Baker94}. 
Another motivation for studying storage capacity on planar graphs is that they represent common network topologies
for distributed systems. For example, see \cite{bowden2011planarity} to note how a surprising number of data networks are actually planar. 
To minimize interference, it is natural for a distributed storage system to be arranged as a planar network. Moreover it is useful to have wireless
networks, video-on-demand networks etc.~that are planar or almost planar. 

Video-on-demand also motivates a related broadcast problem 
called {\em index coding}  \cite{bar2011index} for which planar topologies are of interest, and outerplanar topologies have already
been studied \cite{berliner2011index}.
It was shown in \cite{mazumdar2015storage} that  storage capacity is, in a coding-theoretic sense, dual to 
 index coding and  is  equivalent to the {\em guessing game} problem of \cite{gadouleau2015fixed}.
 Let   $\{X_i\}_{i\in V}$ be independent uniform random variables
 each with a finite sample space of size $q$.
The index coding rate for a graph $G(V,E)$ is defined to be the optimum value  of the following minimization problem: 
\begin{align}
\min H(Y)
\end{align} where $Y$ is a random variable with finite support such that 
$$
H(X_i|Y,X_{N(i)})=0,
$$
 for all $i\in V$. This is called the optimum index coding rate for the graph $G$, and we denote it as $\ind_q(G)$. We can also define,
\begin{align}
 \ind(G) =\inf_{q} \ind_q(G).
\end{align}

The index coding problem is the hardest of all network coding problems and has been the subject of much recent attention, see, e.g., \cite{langberg2011hardness}. In particular it can be shown that  any network coding problem can be reduced to an index coding problem \cite{effros2015equivalence}. 
It has been shown  that (see  \cite{mazumdar2015storage}),
\begin{align}\label{eq:indsc}
n -\ind_q(G) \leq \scap_q(G) \leq  n -\ind_q(G) + \log_q (n\ln q) \ .
\end{align}
From this we claim,
\begin{align}
\scap(G) = n -\ind(G) \ .
\end{align}
To see this, define $\scap^{(q)}(G) \equiv \sup_{m\in \integers_+} \scap_{q^m}(G)$ and $\ind^{(q)}(G) \equiv \inf_{m\in \integers_+} \ind_{q^m}(G)$. It is evident that $(m+r)\scap_{q^{m+r}}(G) \ge m\scap_{q^m}(G)+ r\scap_{q^r}(G)$ and  $(m+r)\ind_{q^{m+r}}(G) \le m\ind_{q^m}(G)+ r\ind_{q^r}(G)$ for any two nonnegative integers $m$ and $r$.
Therefore, using Fekete's lemma, $\scap^{(q)}(G) =\lim_{m\to \infty} \scap_{q^m}(G)$ and $\ind^{(q)}(G) =\lim_{m\to \infty} \ind_{q^m}(G)$. Now this gives, from \eqref{eq:indsc},
\begin{align}
n -\ind^{(q)}(G) \leq \scap^{(q)}(G) \leq  n -\ind^{(q)}(G).
\end{align}
Taking supremum on both sides we have, $\scap(G) = n -\ind(G)$.

Hence, exact computation of $\ind(G)$ and $\scap(G)$ is equivalent although the approximation hardness could obviously differ.  
Note that, 
$$
\ind(G) \ge \alpha(G),
$$
where $\alpha(G)$ is the independence number of $G$. Since, for planar graphs
$\alpha(G) \ge n/4$, taking $Y$ to be $X_{[n]}$ already gives a $4$-approximation 
for the index coding rate for planar graphs since $H(Y)\le n$ \cite{Arbabjolfaei016}. In this paper, we give a significantly better approximation algorithm for index coding rate of planar graphs.
Not only that, due to the relation between index coding rate and storage capacity, we can obtain an approximation factor significantly better than $2$
for storage capacity. 
Note that, for general graphs even to approximate  the optimal index coding rate
within a factor of $n^{1-\epsilon}$ seems to be a challenge \cite{BlasiakKL10}.

To go beyond the realm of planar graphs, and to obtain better approximation ratios, we then develop several upper bounding 
tools for storage capacities. In particular by using these tools, we are able to exactly characterize storage capacities of various families of graphs. Our approach revisits a linear program proposed by Blasiak, Kleinberg, and Lubetzky \cite{BlasiakKL11} that can be used to lower bound the optimum index coding rate or upper bound the storage capacity. We transform the problem of bounding this LP into the problem of constructing a family of vertex covers for the input graph. This in turn allows us to upper bound the storage capacity of any graph that admits a specific type of vertex partition. We then identify various graphs for which this upper bound is tight.

Since, the storage capacity, or the vertex cover, act as absolute upper bounds on the rate of information storage
in the graph, a natural question to ask is, if we store 
above the limit of minimum  vertex cover in the graph, will any of the repair property be left? This is similar in philosophy to the
rate-distortion theory of data compression, where one compresses beyond entropy limit and still can recover the data 
with some distortion. This question gives rise to the notion of recovery from partial failure, as defined below. 

We  define the {\em partial repair capacity} also keeping the application of distributed storage in mind. This is a direct generalization
in the context of distributed storage application to handle
partial failure of vertices. In particular, suppose we lose $0 \leq \delta \leq 1$ proportion of the bits stored in a vertex. We still want to
recover these bits by accessing the remaining $(1-\delta)$-fraction of the bits in the vertex plus the contents of the neighborhood.  
What is the maximum amount of information that can be stored in the network with such restriction? Intuitively,
the storage capacity should increase. We characterize  the trade-off between $\delta$ and this increase in storage capacity
from both sides (i.e., upper and lower bounds on the capacity). A surprising fact that we observe is that, if 
we want to recover from
more than half of the bits being lost, then there is no increase in storage capacity.

In summary, we made progress on the study of storage capacity on three fronts:
\begin{itemize}
\item {\em Planar graphs.} We prove a $3/2$ approximation of absolute storage capacity  and $2$ approximation 
for index coding rate  for planar graphs. We provide an approximation guarantee that depends on the number of triangles
in the graph and, in the special case of triangle-free graphs, we get a $4/3$ approximation for storage capacity, and $3/2$ approximation
for index coding rate.  In addition to this, for a constant-size $q$ alphabet, we give a polynomial-time approximation scheme (PTAS) or an $(1+\epsilon)$ approximation, $\epsilon >0,$
for $\ind_q(G)$ of planar graphs using the well-known planar separator theorem \cite{LiptonT80}. 
\item {\em Tools for finding storage capacity upper bounds.} We develop an approach for bounding storage capacity in terms of a small number of vertex covers. We first illustrate this approach by finding the exact storage capacity of some simple graphs. We then use the approach to show a bound on any graph that admits a specific type of vertex partition. With this we prove exact bounds on a family of Cartesian product graphs and a family closely related to outerplanar graphs.
\item {\em Partial failure recovery.} We show that if recovery from neighbors is possible for up to
$\delta$-proportion failure of the bits stored in a server, then the capacity is upper bounded by the optimum value of
a linear program; in particular this implies when $\delta \ge \frac12$, then the partial recovery capacity is same as the 
storage capacity. For an odd cycle, the upper bound on partial recovery capacity is given by $\frac{n}{2}(1+R_2(\delta))$, where
$R_2(\delta)$ is the maximum achievable rate of a binary error-correcting code with relative minimum  Hamming distance at least 
$\delta n$. On the other hand, we also obtain general lower bounds on the partial recovery capacity of a graph.
For an odd cycle, our results imply that a partial failure recovery capacity of $\frac{n}{2}(2- h_2(\delta))$ is polynomial time achievable, where $h_2(\delta)$ denotes the binary entropy function. Our bounds are likely to be tight, since it is a widely believed conjecture  that $R_2(\delta) = 1- h_2(\delta)$ (the Gilbert-Varshamov bound).
\end{itemize} 

\paragraph{Organization.} The remainder of the paper is organized as follows. In Sections \ref{sec:prelim} and \ref{sec:exact} we state and prove some preliminary algorithmic results regarding the storage capacity and index coding that will be useful in proving subsequent results. In section \ref{sec:planar}, we prove our approximation results for planar graphs that include a $3/2$ approximation for $\scap(G)$, $2$ approximation for $\ind(G)$ and a PTAS for $\ind_q(G)$. In Section \ref{sec:upper}, we show a vertex partition approach to upper bound the storage capacity. The results regarding recovery from partial node failure are described in Section \ref{sec:partial}.



\section{Preliminaries}

\label{sec:prelim}
Let $\fcc(G)$ denote the fractional clique packing of a graph $G(V,E)$ defined as follows: Let $\calC$ be the set of all cliques in $G$. For every $C\in \calC$ define a variable $0\leq x_C \leq 1$. Then  
$\fcc(G)$ is the maximum value of 
\begin{align*}
\sum_{C\in \calC} x_C(|C|-1) \tag{CP($G$)}
\end{align*}
 subject to the constraint  that \[\sum_{C\in \calC: u\in C} x_C\leq 1\qquad \forall u\in V.  \]
Note that $\fcc(G)$ can be computed in polynomial time in graphs, such as planar graphs, where all cliques have constant size. Furthermore, $\fcc(G)$ is at least the size of the maximum fractional matching and they are obviously equal in triangle-free graphs since the only cliques are edges.

The following preliminary lemma shows that $\scap_q(G)\geq \fcc(G)$ for sufficiently large $q$. An equivalent result is known in the context of index coding but we include a proof here for completeness.
The basic idea is that we can store $k-1$ units of information on a clique of size $k$ by assigning $k-1$ independent uniform random variables to $k-1$ of the vertices and setting the final random variable to the sum (modulo $q$) of the first $k-1$ random variables. 

\begin{lemma}\label{lem:scaplb}
$\scap_q(G)\geq \fcc(G)$ for sufficiently large $q$.
\end{lemma} 
\begin{proof}
Let $\{x_C\}_{C\in \calC}$ achieve $\fcc(G)$. Let $q$ be sufficiently large power of $2$ such that ${x_C}\cdot\log_2q  $ is integral for every $C$. For each clique $C=\{u_1,\ldots, u_{|C|}\}$ in the graph, define a family of random variables $X_{u_1}^C, X_{u_2}^C, \ldots, X_{u_{|C|}}^C$ where $X_{u_1}^C, X_{u_2}^C, \ldots, X_{u_{|C|-1}}^C$ are independent and uniform over $\{0, 1, \ldots, q^{x_C}-1\}$ and 
\[
X_{u_{|C|}}^C=\sum_{i=1}^{|C|-1} X_{u_i}^C \bmod q^{x_C} \ .
\]
Note that each $X_{u_i}^C$ can be deduced from $\{ X_{u_j}^C \}_{j\neq i}$ and the entropy of $\{ X_{u_j}^C \}_{j\neq i}$ is $x_C(|C|-1)$. Finally, let $X_u$ be an encoding of $\{X_u^C\}_{C\in \calC: u\in C}$ as a symbol from a $q$-ary alphabet; the fact that this is possible follows because $\sum_{C\in \calC: u\in C} x_C\leq 1$. Then the entropy of $\{X_u\}_{u\in V}$ is $\fcc(G)$ as required.
\end{proof}

\paragraph{Further notations.}

Let $G[S]$ denote the subgraph induced by $S \subseteq V$. Let $\is(G)$ be the size of the largest independent set and let $\vc(G)$ be the size of minimum vertex cover. Let $\mm(G)$ be the size of the largest matching of the graph and let $\fm(G)$ be the weight of the maximum fractional matching, i.e., $\fm(G)$ is the maximum of $\sum_{e\in E} x_e$ such that $\sum_{e: v \in E} x_e \le 1$ for all $v \in V$ and $0 \le x_e \le 1, \forall e \in E.$

\section{Exact algorithms for storage capacity and index coding rate}
\label{sec:exact}
\label{sec:exact}
In this section, we describe some (super)exponential time algorithms that compute $\scap_q(G)$ and $\ind_q(G)$ exactly for a constant $q$. First, we review the approach involving a \textit{confusion graph} and then present another algorithm for storage capacity that has faster running time on some graphs.

Let the vertices of the graph $G$ be represented by $1, 2 \dots, n$. A confusion graph $\mathcal{G}(G)$ of $G$ has $q^n$ vertices. Each of the different $q$-ary strings of length $n$ represent the vertices. There is an edge between two strings $x$ and $y$ if there exists some $i \in \{1,\dots, n\}$, such that $x_i \neq y_i$, but $x_j = y_j$ for all $j \in N(i),$ neighbors of $i$ in $G$. It is known that  the  index coding rate $\ind_q(G)$ is equal to logarithm of the chromatic number of $\mathcal{G}(G)$ and the storage capacity $\scap_q(G)$ is the logarithm of the size of maximum independent set of $\mathcal{G}(G)$ \cite{alon2008,mazumdar2015storage}.
Thus, by employing the best currently known algorithms for computing chromatic number~\cite{bhk2009} and maximum independent set~\cite{xiao2013}, we can compute $\ind_q(G)$ and $\scap_q(G)$ in time $O(2.2461^{q^n})$ and $O(1.1996^{q^n})$ respectively. 


To compute storage capacity $\scap_q(G)$, we can take a different approach. We can go over all the possibilities for the $n$ recovery functions $f_1, \dots, f_n$ corresponding to vertices $1, 2, \dots, n$, and for each of the possibilities construct the corresponding  storage codes. To be precise, let vertex $i\in [n]$ have degree $d(i)$. There are $q^{q^{d(i)}}$ possibilities for the recovery function $f_i$. The total number of possibilities for $n$ recovery functions is then \[\prod_{i=1}^n q^{q^{d(i)}} \leq q^{n q^{\Delta}} \ , \]
where $\Delta$ is the maximum degree in $G$.  For each possibility, we go over the $q^n$ vectors and remove the ones that do not agree with the recovery functions, which takes time $O(m)$, where $m=\frac{1}{2} \sum_{v \in V(G)} d(v)$ is the number of edges in $G$, per vector. The remainder of the vectors form the storage code that agrees with the recovery functions. We find the combination of recovery functions that give the largest storage code. The total running time of this algorithm is at most $O(m q^n q^{n q^{\Delta}})$. Since without loss of generality we can assume that the largest storage code contains the all-zero vector, the running time can slightly be improved to $O(m q^{n q^{\Delta}})$. If $\Delta$ is small then this approach is faster than building a confusion graph and computing maximum independent set on it. 

To summarize, we have the following result that will be used to design a PTAS for $\ind_q(G)$ for planar graphs.
\begin{theorem}
Computing $\scap_q(G)$ exactly takes time \[O\left( \min \left\{ 1.1996^{q^n},  m q^{n q^{\Delta}} \right\}\right)\] for a graph with $n$ vertices, $m$ edges, and maximum degree $\Delta$. Computing $\ind_q(G)$ exactly takes time $O(2.2461^{q^n})$.
\end{theorem}

\section{Approximation algorithms for planar graphs via vertex cover}
\label{sec:planar}
\label{sec:planar}
In this section, we present approximation results for the storage capacity and optimal index coding rate of planar graphs. Specifically we show that $\fcc(G)$ can be used to achieve a $3/2$ approximation of the storage capacity and a $2$ approximation of the optimal index coding rate. 


In our storage capacity result we use ideas introduced by Bar-Yehuda and Even\cite{Bar-YehudaE82} for the purpose of $5/3$-approximating the vertex cover in planar graphs. Specifically, they first considered a maximal set of vertex-disjoint triangles, reasoned about the vertex cover amongst these triangles, and then reasoned about the triangle-free induced subgraph on the  remaining vertices. We consider a similar decomposition and reason about the integrality gap of vertex cover in each component. We parameterize our result in terms of the number of triangles; this will be essential in the subsequent result on optimal index coding rate.
 
\begin{theorem}\label{thm:scap}
Assume $G$ is planar and let $T$ be a set of $3t$ vertices corresponding to maximal set of $t$ vertex disjoint triangles.  Then,
\[1 \leq \frac{\scap(G)}{\fcc(G)} \leq \frac{3t+k}{2t+3k/4}\]
where $k$ is the size of the minimum vertex cover of $G[V\setminus T]$. Hence $\fcc(G)$ is a $3/2$ approximation for $\scap(G)$ and $4/3$ approximation if $G$ is triangle-free.
\end{theorem}
\begin{proof}
Let $G'=G[V\setminus T]$. Partition the set of vertices $V$ into $T\cup C \cup I$ where $C$ is the minimum vertex cover of $G'$ and $I\subset V\setminus T$ is therefore an independent set. Let $X_V$ be the set of variables that achieve storage capacity. Therefore, 
\[
\scap(G)=H(X_V)=H(X_T)+H(X_C|X_T)+H(X_I|X_C,X_T)\leq H(X_T)+H(X_C|X_T)\leq 3t+k
\]
since for each $v\in I$, $H(X_v|X_C,X_T)=0$ since $N(v)\subset C\cup T$.

Consider the fractional clique packing in which each of the $t$ vertex-disjoint triangles in $T$ receive weight $1$. Then,
$\fcc(G)\geq 2t+\fcc(G')$.
Then it remains to show that $\fcc(G')\geq 3k/4$. 
%
Note that since $G'$ is triangle-free planar graph, it is 3-colorable by Gr\"otzsch's theorem \cite{grotzsch59}. Furthermore, $\fcc(G')$ is the maximum fractional matching which, by duality, is the minimum fractional vertex cover. Hence it suffices to show that the size of the minimum fractional vertex cover of a 3-colorable graph is at least $3/4$ of the size of the minimum (integral) vertex cover, i.e., $3k/4$. This can be shown as follows. Let  $x_1, \ldots, x_n$ be an optimal fractional vertex cover, i.e., for all edges $uv\in G'$, $x_u+x_v\geq 1$. Since fractional vertex cover is $1/2$-integral (i.e., there exists an optimal solution where each variable is either integral or $1/2$)~\cite{hochbaum1982approximation}, we may assume each $x_u\in \{0,1/2,1\}$. Let $I_1,I_2,I_3$ be a partitioning of $\{u\in [n]:x_u=1/2\}$ corresponding to a 3-coloring where 
\[\sum_{v\in I_1} x_v\geq \sum_{v\in I_2} x_v\geq \sum_{v\in I_3} x_v \ .\] Then consider $y_1, \ldots, y_n$ where $y_u=1$ iff $u\in I_2\cup I_3$ or $x_u=1$. Then 
\[
\sum_{u\in [n]} y_u\leq \sum_{u\in I_2\cup I_3} y_u
+ \sum_{u\in [n]:x_u=1} y_u
\leq 2/3\cdot 2 \cdot \sum_{u: x_u=1/2} x_u+ \sum_{u: x_u=1} x_u\leq 4/3 \cdot \fcc(G') \ ,
\]
and $y_1, \ldots, y_n$ is a vertex cover because for every edge $uv$, at least one endpoint one of $\{x_u,x_v\}$ is 1 or at least one of $u$ and $v$ is in $I_2\cup I_3$. 
%
\end{proof}

We next use the result of the previous theorem, together with the chromatic number of planar and triangle-free planar graphs to achieve a $2$ approximation for $\ind(G)$.

\begin{theorem} \label{thm:ind}
Assume $G$ is planar and let $T$ be a set of $3t$ vertices corresponding to $t$ vertex disjoint triangles.  Then,
\[1 \leq \frac{n-\fcc(G)}{\ind(G)} \leq \begin{cases} 
\frac{3n+3t}{4n-12t}+\frac{3}{4} & \mbox{for } t \leq \frac{n}{12} \\
\frac{t}{n}+\frac{7}{4} & \mbox{for } \frac{n}{12} \leq t \leq \frac{n}{4} \\
4 - \frac{8t}{n} &  \mbox{for } t \geq \frac{n}{4}\end{cases} 
\ . \]
Maximizing over $t$ implies that $n-\fcc(G)$ is a $2$ approximation for $\ind(G)$ and a $3/2$ approximation if $G$ is triangle-free.
\end{theorem}


\begin{proof}
From Theorem~\ref{thm:scap}, we know that $\fcc(G) \geq 2t + 3/4 \vc(G')$ where $G'=G[V\setminus T]$. Therefore, we can bound $n-\fcc(G)$ as follows:
\begin{align*}
n-\fcc(G)  & \leq n - (2t + 3/4 \vc(G')) \\
&= n - (2t + 3/4(n-3t-\is(G')))\\
&= (n+t)/4 + 3/4\is(G')
\end{align*}
On the other hand,
\[\ind(G)  \geq \is(G) \geq \is(G') \]
where $\is$ denotes the size of the maximum independent set of the graph.
Note that $\is(G) \geq n/4$ since $G$ is planar and thus 4-colorable \cite{appel1977b,appel1977}. Since $G'$ has $n-3t$ vertices and is triangle-free and planar and thus 3-colorable  \cite{grotzsch59}, 
$n-3t \geq \is(G') \geq (n-3t)/3$.
By combining inequalities above we get
\begin{eqnarray*}
\frac{n-\fcc(G)}{\ind(G)} 
& \leq &  
 \min \left (\frac{(n+t)/4 + 3/4\is(G')}{\is(G)}, \frac{(n+t)/4 + 3/4\is(G')}{\is(G')}, \frac{(n+t)/4 + 3/4\is(G)}{\is(G)} \right ) \\
& \leq & 
 \min \left (\frac{(n+t)/4 + 3/4(n-3t)}{n/4}, \frac{(n+t)/4}{(n-3t)/3} + 3/4, \frac{(n+t)/4}{n/4} + 3/4 \right ) \\
&= & \min \left (4 - \frac{8t}{n},\frac{3n+3t}{4n-3t} + \frac{3}{4}, \frac{t}{n}+\frac{7}{4} \right ) 
\ .
\end{eqnarray*}
\end{proof}

\subsection{Approximation algorithms for fixed alphabet size}

In this section, we present improved polynomial time algorithms for approximating the storage capacity and the index coding rate on the assumption that the alphabet size $q$ of the codes is a fixed constant. These are based on combining the exact algorithm described in Section \ref{sec:exact} with the planar separator theorem \cite{LiptonT80}.

\begin{lemma}\label{lem:removenodes}
Let $G'$ be a graph formed by removing at most $k$ vertices from $G$. Then,
\[ \scap_q(G)-k \leq \scap_q(G')\leq \scap_q(G) \ .\]
\end{lemma}
\begin{proof}
Consider an optimal storage code and fix the value of the $k$ vertices to their most common $k$-tuple of values. This decreases the size of the code by at most a factor $q^k$ and hence decreases the storage capacity by at most $k$. The fact that $\scap_q(G')\leq \scap_q(G)$ follows because any code for $G'$ is also a code for $G$.
\end{proof}

The following lemma shows that the removal of a relatively small number of vertices in a planar graph is sufficient to ensure that the remaining components are small. The proof is  essentially the same as that used to establish Theorem 3 in \cite{LiptonT80}. 

\begin{lemma}\label{lem:planarsep}
There exists $n_0\leq \epsilon n$ vertices $V_0$ in a planar graph whose removal yields a graph with at most $r=O(\epsilon^{2} n)$ components $G_1, G_2, \ldots , G_r$ where each $G_i$ has $n_i=O(1/\epsilon^2)$ vertices and there are no edges between $G_i$ and $G_j$.  Furthermore,  $V_0$, $G_1, G_2,\ldots$ can by found in polynomial time.
\end{lemma}
\begin{proof}
The planar separator theorem \cite{LiptonT80} states that in any planar graph there exists a set of 
at most $c \sqrt{n}$ vertices (for some large constant $c$) such that the removal of these vertices disconnects the graph into components of size at most $2n/3$. Furthermore, these vertices can be found in polynomial time. Consider applying  this recursively on all components until they each have size less than $\alpha c^2/\epsilon^2$ for a sufficiently large constant $\alpha>1$ . Then, the total number of vertices removed is at most
\begin{eqnarray*}
|V_0| 
&\leq & c \sqrt{n} \frac{n}{(2/3)n}+
 c \sqrt{(2/3)n} \frac{n}{(2/3)^2n}+
 c \sqrt{(2/3)^2n} \frac{n}{(2/3)^3n}+ \ldots 
\ldots
+  c \sqrt{(3/2) \alpha c ^2/\epsilon^2} \frac{n}{\alpha c^2/\epsilon^2}
\\
& \leq &
  \epsilon n
  \end{eqnarray*}
because, during the course of the recursion, the planar separator theorem is applied to at most $n/((2/3)^i n)$ graphs of size between $(2/3)^i n$ and $(2/3)^{i-1} n$ and each such application involves the removal of at most $c\sqrt{(2/3)^{i-1} n}$ vertices.
%
%
%
%
\end{proof}

\begin{theorem}
There exists polynomial time algorithms to  approximate $\ind_q(G)$ up to a multiplicative factor of $1+\epsilon$ and  approximate $\scap_q(G)$ up additive error $\epsilon n$.
\end{theorem}
\begin{proof}
The result for $\scap_q(G)$ follows by applying the decomposition in Lemma \ref{lem:planarsep} to $G$ and then solving the problem optimally on $G_1, G_2,\ldots$ using the algorithm in Section \ref{sec:exact}.  Let $G'=G_1\cup G_2 \cup \ldots$. By Lemma \ref{lem:removenodes},
\[
\scap_q(G)-n_0 \leq \scap_q(G')= \scap_q(G_1)+\scap_q(G_2) +  \ldots
\]
and hence finding the optimal solution for each $\scap_q(G_i)$ yields an additive $n_0\leq \epsilon n$ approximation as required.

The result for $\ind_q(G)$ is similar. By using the results in   Section \ref{sec:exact}, we can find the optimal index coding solution of each $G_i$. Combining these with the naive solution for the $n_0$ vertices in $V_0$ yields a solution with rate $n_0+\sum_{i\geq 1} \ind_q(G_i)$.
We can relate $\ind_q(G)$ to this as follows:
%
%
\begin{eqnarray*}
n_0+\sum_{i\geq 1} \ind_q(G_i)
& \leq & 
n_0 +\sum_{i\geq 1} \left ( n_i -\scap_q(G_i) + \log_q (n_i\ln q) \right ) \\
&\leq &  n_0+n-n_0 -\scap_q(G') + r \log_q (\max_i n_i\ln q) \\
&\leq &  n_0+ \ind_q(G') +r \log_q (\max_{i\geq 1} n_i\ln q) \\
&\leq &  \ind_q(G) +2n_0+r \log_q (\max_{i\geq 1} n_i\ln q)\\
&= &  \ind_q(G) \left (1+\frac{2n_0+r \log_q (\max_{i\geq 1} n_i\ln q)}{\ind_q(G)}\right ) \\
&= &  \ind_q(G) \left (1+8\epsilon +O(\epsilon^{2} \log_q (\epsilon^{-2}\ln q)) \right ) \\
\end{eqnarray*}
where the last line follows since $\ind_q(G)\geq n/4$, $n_0\leq \epsilon n$, $r=O(\epsilon^{2}n), n_i=O(\epsilon^{-2})$. Reparamaterizing by $\epsilon\leftarrow \epsilon/c$ for some sufficiently large constant gives the required result.
\end{proof}


\section{Upper Bounds on Storage Capacity via Multiple Vertex Covers}
\label{sec:upper}
In this section, we start by considering a linear program proposed by Blasiak, Kleinberg, and Lubetzky \cite{BlasiakKL11} that can be used to lower bound the optimum index coding rate or upper bound the storage capacity.\footnote{Blasiak et al.~only consider index coding but it is straightforward to adapt the LP in a natural way for storage capacity; see below.} Unfortunately there are $\Omega(2^n)$ constraints but by carefully selecting a subset of constraints we can prove upper bounds on the storage capacity for a specific graph without solving the LP. 

Our main goal in this section is to relate this linear program to finding a suitable family of vertex covers of the graph. In doing so, we propose a  combinatorial ``gadget" based approach to constructing good upper bounds that we think makes the process of proving strong upper bounds more intuitive. This allows us to prove a more general theorem that gives an upper bound on the storage capacity for a relatively large family of graphs. As an application of this theorem we show that a class of graphs closely related to the family of outerplanar graphs and another family of cartesian product graphs have capacity exactly $n/2$.

%
%
%

\subsection{Upper Bound via the ``Information Theoretic" LP}\label{sec:infoLP}
We first rewrite the index coding LP proposed by Blasiak, Kleinberg, and Lubetzky \cite{BlasiakKL11} for the purposes of upper-bounding  storage capacity. We define a variable $z_S$ for every $S \subseteq V$ that will correspond to an upper bound for $H(X_S)$. Let $\cl(S) = S \cup \{v : N(v) \subseteq S \}$ denote the \textit{closure} of the set $S$ consisting of vertices in $S$ and vertices with all neighbors in $S$.
\begin{align*}
\textrm{maximize}\quad & z_V \tag{Information theoretic LP}\\
\textrm{s.t.}\quad & z_{\emptyset} = 0\\
& z_T - z_S \leq  |T \setminus \cl(S)| \quad \forall S \subseteq T \\
& z_S + z_T \geq z_{S \cap T} + z_{S \cup T} \quad \forall S,T
\end{align*}
The second constraint corresponds to 
\[H(X_T)-H(X_S)=H(X_T| X_S) =H(X_T| X_{\cl(S)})\leq H(X_{T\setminus \cl(S)})\leq |T\setminus \cl(S)| \ , \]
whereas the last constraint follows from the sub-modularity of entropy. Hence, the optimal solution to the above LP is an upper bound on $\scap(G)$. We henceforth refer to the above linear program as the \emph{information theoretic} LP.

\subsection{Upper Bound via Gadgets}

\textit{$k$-cover by gadgets} is a technique for proving upper bounds  on the storage capacity of a graph. The core idea is to construct a set of $k$ vertex covers for the graph via the construction of various gadgets which we now define.


A gadget $g = (S_1, S_2, c_1, c_2)$ consists of two special sets of vertices $S_1$ and $S_2$ and two colors $c_1$ and $c_2$. Sets $S_1$ and $S_2$ are created in the following way: take two sets of vertices $A$ and $B$, take their closures $\cl(A)$ and $\cl(B)$, then $S_1 = \cl(A) \cup \cl(B)$ and $S_2 = \cl(A) \cap \cl(B)$. Call $S_1$ the outside of the gadget and $S_2$ its inside, as it is always the case that $S_1 \supseteq S_2$. We note that by taking $A=\{v\}$ and $B=\emptyset$ we obtain a gadget with the outside $\{v\}$ and empty inside (assuming $v$ has no neighbors of degree one); call such gadget \textit{trivial}. Define the weight of a gadget to be $w(g) = |A|+|B|$. As for the two colors of the gadget, they are picked from a fixed set of $k$ colors; $c_1$ is assigned to all vertices in $S_1$ and $c_2$ is assigned to all vertices in $S_2$. Note, that a vertex can be assigned multiple colors in that manner. The objective is to find a set of gadgets, such that for any color $c$, vertices with $c$ assigned to them form a vertex cover. If that condition is met for all $k$ colors, we show that the total weight of the gadgets used provides an upper bound on $k \scap(G)$.


We can formulate the $k$-cover by gadgets (for fixed $k$) as the following integer linear program:
Let $x_{g,S}$ be a variable where vertex set $S$ is a part of gadget $g$. Note that there are two variables per gadget, corresponding to its outside and inside sets. $x_{g,S} = 1$ if $g$ participates in the cover by gadgets and 0 otherwise. We use $c_g(S) = c$ to denote that the color of vertex set $S$ in gadget $g$ is $c$.
\begin{align*}
\textrm{minimize}\quad &\frac{1}{k}\sum_{g, S} x_{g, S} \cdot w(g)/2 \tag{k-cover by gadgets} \\
\textrm{ s.t.}\quad& \sum_{\substack{g, S : u \in S, \\ c_g(S)=c}} x_{g, S} + \sum_{\substack{g', S' : v \in S', \\ c_{g'}(S')=c}} x_{g', S'} \geq 1 \quad \forall (u,v) \in E, \forall c \\
&x_{g, S} = x_{g, S'} \quad \forall g \textrm{ with outside $S$ and inside $S'$}
\end{align*}
The first condition states that the union of sets of a certain color is a vertex cover and the second one states that both the inside and outside of $g$ participate (or do not participate) in the gadget cover.

\begin{theorem}
Any feasible integral solution to the above $k$-cover by gadgets integer LP is an upper bound on $\scap(G)$.
\end{theorem}
\begin{proof}
We prove this by considering the information theoretic LP above and showing how to bound $z_V$ by the gadget weights.
First, note which constraints correspond to the steps of creating a gadget:
\begin{align*}
z_A \leq |A| \quad & \textrm{take set $A$}\\
z_B \leq |B| \quad & \textrm{take set $B$}\\
z_{\cl(A)} - z_A \leq 0 \quad  &\textrm{find closure of $A$}\\
z_{\cl(B)} - z_B \leq 0 \quad & \textrm{find closure of $B$}\\
z_{S_1} + z_{S_2} \leq z_{\cl(A)} + z_{\cl(B)} \quad & \textrm{find $S_1 = \cl(A) \cup \cl(B)$ and $S_2 = \cl(A) \cap \cl(B)$}
\end{align*}
If we sum all the constraints, we obtain $z_{S_1} + z_{S_2} \leq |A|+|B|$. Recall that $|A|+|B|$ is the weight of the gadget. Let $H$ be a cover by gadgets. Then by summing all corresponding constraints, we get
$$\sum_{g=(S_1, S_2, c_1, c_2) \in H} z_{S_1} + z_{S_2} \leq \sum_{g \in H} w(g)$$
Group the sets participating in gadgets in $H$ into color classes $C_1, C_2, \dots, C_k$. Let $U_i = \bigcup\limits_{S \in C_i}S$ be the set of all vertices of color $c_i$. The corresponding constraints are then
\begin{align*}
z_{U_i} - \sum_{S \in C_i} z_S \leq 0 \quad & \forall i \in \{1,2,\dots, k\}\\
z_{cl(U_i)} - z_{U_i} \leq 0 \quad & \forall i \in \{1,2,\dots, k\}
\end{align*}
Note that $z_V = z_{cl(U_i)}$ since $U_i$ is a vertex cover. By summing these $2k$ constraints and the one obtained from building gadgets, we get $k z_V \leq \sum_{g \in H} w(g)$.
\end{proof}

\subsubsection{Examples}
We next illustrate the use of the $k$-cover via gadgets approach with a couple of examples. First, we use a 2-cover by gadgets  to re-prove a result of Blasiak et al.~\cite{BlasiakKL10} that established that the storage capacity of a cycle on $n$ vertices is $n/2$. The point of this first example is to illustrate simplicity of the new approach.  Then we give an example of an outerplanar graph for which we establish a tight bound of 14/3. In addition to serving as another example of the new approach, we think this example is particularly interesting as it demonstrates that it is sometimes necessary to consider a $k$-cover via gadgets where $k>2$ in order to establish a tight result. Note that any upper bound via a $k$-cover via gadgets is a multiple of $1/k$ and hence $k=1$ or $k=2$ would be insufficient to prove a tight bound of 14/3.

\begin{figure}
\begin{center}
\subfigure[\small{An Odd Cycle.}]{~~~~\includegraphics[scale=0.9]{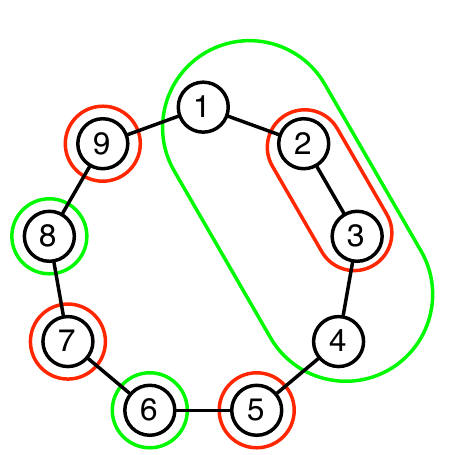}~~~~\label{fig:cycle}}
\subfigure[\small{An Outerplanar Graph.}]{~~~~\includegraphics[scale=0.7]{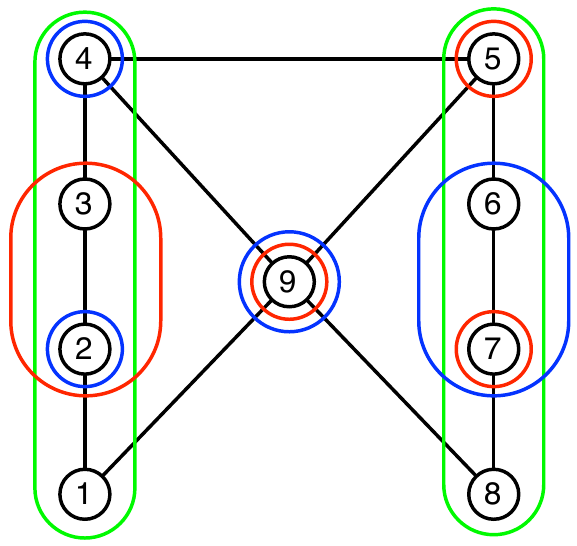}~~~~\label{fig:outer}} ~~
\caption{Two examples of $k$-cover upper bounds. See text for details.}
\end{center}
\end{figure}

\paragraph{Odd Cycles.}
We prove that the storage capacity of an odd cycle of length $n$ is $n/2$; see Figure~\ref{fig:cycle} for an example where $n=9$.  $\fm(C_n) = n/2$, thus $\scap(C_n) \geq n/2$. For the upper bound we create a gadget by taking $A = \{v_1, v_3\}$, $B = \{v_2, v_4\}$ and obtaining outside set $S_1 = \{v_1, v_2, v_3, v_4\}$ and inside set $S_2 = \{v_2, v_3\}$. On the rest of the vertices we place trivial gadgets. Color $S_1$ and trivial gadgets on $v_6, v_8, \dots, v_{n-1}$ green, color $S_2$ and trivial gadgets on $v_5, v_7, \dots, v_{n}$ red. Green and red sets are then vertex covers and the total weight of all gadgets is $n$. Thus, $\scap(C_n) \leq n/2$.

\paragraph{An Outerplanar Graph.} 
We prove that the storage capacity of the graph in Figure~\ref{fig:outer} is $14/3$. This capacity is achieved by the fractional clique cover. Create gadget $g_1$ from $A_1 = \{v_1,v_3\}$ and $B_1 = \{v_2,v_4\}$ and another gadget $g_2$ from $A_2 = \{v_5,v_7\}$ and $B_2 = \{v_6,v_8\}$. Place one trivial gadget on each of the vertices $v_2, v_4, v_5, v_7$ and 2 trivial gadgets on $v_9$. The sets are colored as follows:
\begin{itemize}
\item Red: $v_5, v_7, v_9$ and the inside of gadget $g_1$
\item Blue: $v_2, v_4, v_9$ and the inside of gadget $g_2$
\item Green: the outside sets of both $g_1$ and $g_2$
\end{itemize}
Note that vertices of every color class form a vertex cover and the total weight of gadgets is $14$.



\subsection{$n/2$ Upper Bound via Vertex Partition}

The next theorem uses a $2$-cover by gadgets to prove that a certain family of graphs have capacity at most $n/2$.  Subsequently, we will use this theorem to exactly characterize the capacity of various graph families  of interest.

\begin{theorem} \label{partition}
Suppose that the vertices of a graph $G$ can be partitioned into sets $X$ and $Y$ such that:
\begin{enumerate}
\item $G[X]$ and $G[Y]$ are both bipartite.
\item $S_X$ is an independent set in $G[X]$ and $S_Y$ is an independent set in $G[Y]$ 
\end{enumerate}
where $S_X\subseteq X$ consists of all vertices in $X$ with a neighbor in $Y$ and 
$S_Y\subseteq Y$ consists of all vertices in $Y$ with a neighbor in $X$. 
Then $\scap(G) \leq n/2$.
\end{theorem}

%
 \begin{proof}
We prove this theorem by showing that $G$ has a 2-cover by gadgets of total weight $n$. Let $(A,B)$ be a bipartition of the vertices of $G[X]$. Create the sets of gadget $g_X$ from $A$ and $B$. Note that the vertices in $X$ that are in the outside set of the gadget but not in the inside set, are exactly the vertices in $S_X$. This follows because for $v\in X$, $v\in \cl(A)\cap \cl(B)$ iff all of $v$'s neighbors are in $X$.
Similarly, create $g_Y$. Color the inside of $g_X$ and the outside of $g_Y$ red. Color the outside of $g_X$ and the inside of $g_Y$ blue. Observe that both color classes are vertex covers and the total weight of the 2 gadgets is $|X|+|Y| = n$.
 \end{proof}
 
 \subsubsection{Cartesian Product of a Cycle and a Bipartite Graph}
We now illustrate an application of Theorem  \ref{partition}. The Cartesian product of graphs $G_1 = (V_1, E_1)$ and $G_2 = (V_2, E_2)$ is denoted by $G_1 \square G_2$ and defined as follows:
\begin{itemize}
\item The vertex set is the Cartesian set product $V_1 \times V_2$
\item $(u,u')(v,v')$ is an edge iff $u=v$ and $u'v' \in E_2$ or $u'=v'$ and $uv \in E_1$
\end{itemize}

We next use Theorem  \ref{partition} to show that any Cartesian Product of a cycle and a bipartite graph has storage capacity exactly $n/2$. An example of such a graph is given in Fig.~\ref{fig:prism} where the bipartite graph considered is just a length 3 path.

\begin{figure}
\begin{center}
\subfigure[\small{Graph $G$. Shaded vertices are $X=S_X$.}]{~~~~~~~~~~~\includegraphics[scale=0.8]{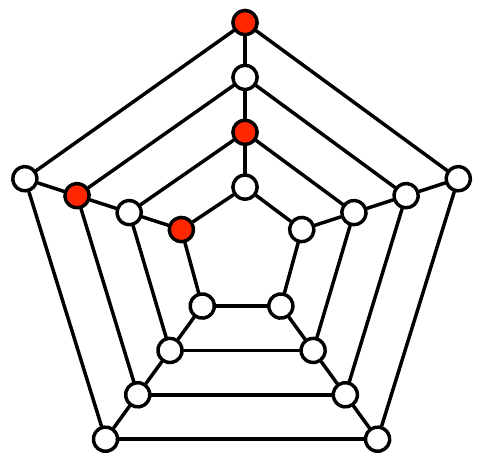}~~~~~~~~~~~\label{fig:prism1}} ~~~~~
\subfigure[\small{Graph $G[Y]$. Shaded vertices are $S_Y$.}]{~~~~~~~~~~~\includegraphics[scale=0.8]{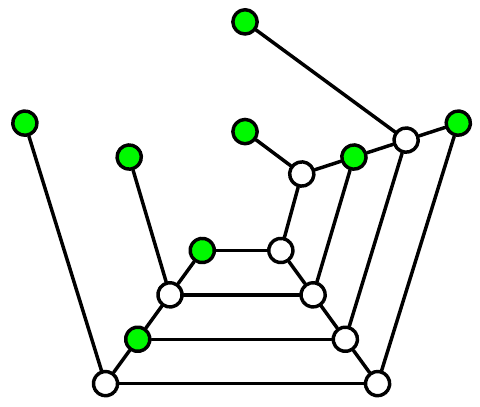}~~~~~~~~~~~\label{fig:prism2}}
\caption{Example of storage capacity proof for a cartesian product graph $G$ formed from a 5-cycle and a length 3 path. See text for details.}\label{fig:prism}
\end{center}
\end{figure}

%

 \begin{theorem}
Let $C_k$ be a cycle with $k>3$, $B$ a bipartite graph, and $G = C_k \square B$. Then $\scap(G) = n/2$, where $n$ is the number of vertices in $G$.
 \end{theorem} 
 
 \begin{proof}
 If $k$ is even, $G$ is a bipartite graph with $\mm(G) = \vc(G) = n/2$ and hence $\scap(G)=n/2$. Assume for the rest of the proof that $k$ is odd.  
 
 To show that $\scap(G)\geq n/2$, consider the fractional matching where we assign weight $1/2$ to all edges of the form $(u,a)(v,a)$, i.e., edges that come from the cycle. Hence  $\scap(G)\geq n/2$.
 
To show that $\scap(G)\leq n/2$ we proceed as follows. Consider the subgraph $G_i$ induced by vertices $(u_i, v_1)$, $(u_i, v_2)$, $(u_i, v_3)$, etc., which is isomorphic to $B$. Fix a bipartition $(R,Q)$ of $B$ and split the vertices of each $G_i$ into $R_i$ and $Q_i$ according to that bipartition. We now show that $X = R_1 \cup Q_2$ and $Y = V \setminus X$ satisfy the conditions of Theorem \ref{partition}. $G[X]$ has no edges and therefore is bipartite. $G[Y]$ is bipartite because it consists of $P_{k-3} \square B$ which is bipartite (where $P_{k-3}$ is a path of length $k-3$ obtained by deleting edges $u_k u_1$, $u_1 u_2$, and $u_2 u_3$ from the cycle), edges between $R_k$ and $R_1$, and edges between $Q_2$ and $Q_3$ which do not complete any cycles. $S_X = X$ is an independent set. $S_Y = R_k \cup Q_1 \cup R_2 \cup Q_3$ is also an independent set.
 \end{proof}
 

 \subsubsection{Cycles With Chords That Are Not Too Close Together}
 
 We next apply Theorem \ref{partition} to prove that a family of graphs related to outerplanar graphs also has storage capacity $n/2$. Recall that any (connected) outerplanar graph without cut vertices is a cycle with non-overlapping chords. The family of graphs we consider is more general in the sense that we permit the chords to overlap but more restrictive in the sense that we require the endpoints of these chords to be at least a distance 4 apart on the cycle. A natural open question is to characterize $\scap(G)$ for all outerplanar graphs. All that was previously known is that if we assume each $X_i$ is a linear combination of $\{X_j\}_{j\in N(i)}$, then $\scap(G)$ equals  \emph{integral} clique packing \cite{berliner2011index}.
 
  \begin{theorem}
Let $G$ be a cycle with a number of chords such that endpoints of chords are at least distance 4 apart on the cycle. Then $\scap(G) = n/2$.
 \end{theorem} 
 
 \begin{proof}
To show that $\scap(G)\geq n/2$, consider the fractional matching where we place weight $1/2$ on every edge of the cycle.
To show that $\scap(G)\leq n/2$ we proceed as follows. Label the vertices that are endpoints of chords $c_1, c_2, \dots, c_k$ in the order they appear on the cycle. For every path between $c_i$ and $c_{i+1}$ (and between $c_k$ and $c_1$) pick the middle vertex of the path to be included in $X$. If the path is of odd length, pick either of the 2 middle vertices. We now show that $X$ and $Y=V \setminus X$ satisfy the conditions of Theorem \ref{partition}. $X = S_X$ is an independent set. $G[Y]$ is a forest and $S_Y$ is an independent set due to the assumption on the distance between chord endpoints.
 \end{proof}


\begin{figure}
\begin{center}
\subfigure[\small{Graph $G$. Shaded vertices are $X=S_X$.}]{~~~~~~~~\includegraphics[scale=0.75]{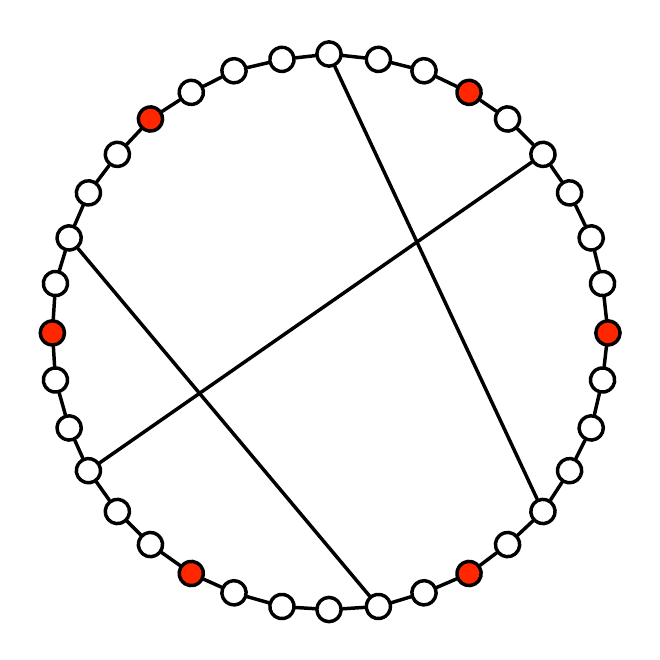}~~~~~~~~\label{fig:chord1}} ~~
\subfigure[\small{Graph $G[Y]$. Shaded vertices are $S_Y$.}]{~~~~~~~~\includegraphics[scale=0.75]{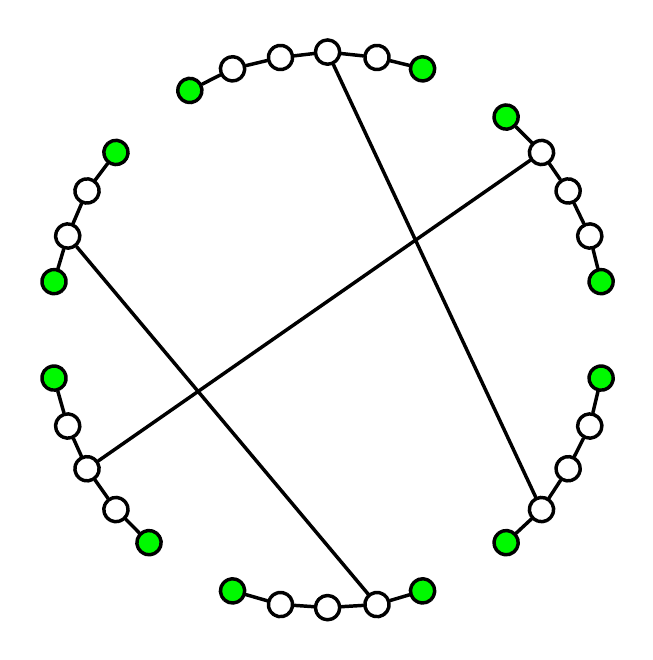}~~~~~~~~\label{fig:chord2}}
\caption{Example of a storage capacity proof for cycles with chords. See text for details.}\label{fig:chords}
\end{center}
\end{figure}

\section{Recovery from partial node failure}
\label{sec:partial}
In this section, we extend the notion of storage capacity to cover for partial failures. This is a new generalization, that, as far as we understand,
does not have a counterpart in index coding.
As before, suppose we have a graph $G(V=[n],E)$ on $n$ vertices. We assume here that   vertex $i \in [n]$ stores $X_i \in \ff_q^m$,  a $q$-ary random vector of length $m$. 
We want the following repair criterion to be satisfied:
if up to any $\delta, 0 \le \delta \le 1,$ proportion of the $m$ coordinates of $X_i, i \in [n]$ are erased, they can be recovered by using the 
remaining content of the vertex $i$ and $X_{N(i)}$, the contents in the neighbors of the vertex.

We call the normalized asymptotic maximum total amount of information (in terms of $q$-ary unit)  
$$\lim_{m\to \infty}  \frac{H(X_1, X_2, \dots, X_n)}m,$$ 
that can be stored in the graph $G$, to be the {\em partial recovery capacity} of $G$.  This is denoted by
$\scap_q(G,\delta)$.

We have the following simple facts. Recall, $\scap^{(q)}(G) \equiv \sup_{m\in \integers_+} \scap_{q^m}(G)$.
\begin{proposition}
For a graph $G$,
$\scap_q(G,0) = n$ and $\scap_q(G,1) = \scap^{(q)}(G)$.
\end{proposition}
\begin{proof}
The first statement is  quite evident. For the second, note that,
$$
\scap_q(G,1) = \lim_{m\to \infty} \scap_{q^m}(G) =  \sup_m \scap_{q^m}(G) = \scap^{(q)}(G),
$$
where we could use the $\lim$ and the $\sup$ interchangeably because of Fekete's lemma, as discussed in the introduction.
\end{proof}

In the remaining parts of this section,   we will provide tight upper and lower bound on 
the quantity $\scap_q(G,\delta)$.

\subsection{Impossibility bound}

Note that, the partial recovery capacity can be defined in terms of an entropy maximization problem, generalizing the 
storage capacity.
\begin{theorem}\label{thm:partial}
Let $H(X)$ be the entropy of $X$ measured in $q$-ary units. Suppose,  $X_i \in \ff_q^m$, $i \in [n]$.
For a graph $G([n], E)$, $\scap_q(G,\delta)$ is upper bounded by the solution of the following optimization problem.
\begin{align}
\max \lim_{m \to \infty} \frac{H(X_1, \dots, X_n)}{m},
\end{align}
such that,
$$
H(X_i \mid X_{N(i)}) \le \log_q A_q(m,\delta m+1),
$$
where $A_q(m,d)$ is the maximum possible size of a $q$-ary $m$-length error-correcting code with minimum distance $d$.
\end{theorem}

\begin{proof}
Let $X_i \in \ff_q^m$, $i \in [n]$ be the random variables that can be stored in the vertices of $G$ satisfying the repair condition. Suppose we are given the values of $X_{N(i)}$. In this situation let $M \subseteq \ff_q^m$ be the set of possible 
values of $X_i$ ($P(X_i =a) >0, \forall a \in M$). Let $X_i^{1}, X_i^{2}$ be any two different elements of $M$. We claim that, 
the Hamming distance between  $X_i^{1}, X_i^{2}$ is at least $\delta m +1$, or 
$$
d(X_i^{1}, X_i^{2}) \ge \delta m +1.
$$
Suppose this is not true. Then there exist $X_i^{1}, X_i^{2} \in M$
such that $d(X_i^{1}, X_i^{2}) \le \delta m.$ Let $J \subset \{1,\dots, m\}$ be the coordinates where $X_i^{1}$ and $X_i^{2}$
differ.  Therefore, $|J| \le \delta m$. Suppose $X_i^{1}$ was stored in vertex $i$ and  the coordinates in $J$ are erased.   Now, 
there will not be any way to uniquely
identify $X_i$: it can be either of  $X_i^{1}$ or $X_i^{2}$. Hence the repair condition will not be satisfied which is a contradiction. 

Therefore, $M \subseteq \ff_q^m$ is a set of vectors such that any two elements of $M$ is Hamming distance at least $\delta m +1$ apart.
Hence $M$ is an error-correcting code with minimum distance $\delta m +1$. And therefore, $|M| \le A_q(m, d)$. This implies,
$
H(X_i \mid X_{N(i)}) \le \log_q A_q(m,\delta m+1),
$ 
which proves the theorem.
\end{proof}

Let us define $$R_q(\delta) \equiv \lim_{m \to \infty} \frac{\log_q A_q(m,\delta m +1)}{m},$$ assuming the limit exists.
\begin{corollary}
 We must have, for any graph $G$,  $\scap_q(G,\delta) = \scap^{(q)}(G)$ for $\delta \ge 1-\frac1q$. In particular, $\scap_2(G,\delta) = \scap^{(2)}(G)$ for $\delta \ge \frac12.$
\end{corollary}
 The proof of this fact follows since $R_q(\delta) =0$ for $\delta \ge 1-\frac1q$ (Plotkin bound, see~\cite[p.~127]{roth2006introduction}).
 
%


Generalizing the technique of upper bounding the storage capacity via an information theoretic linear program, we can obtain an upper bound on $\scap_q(G,\delta)$.
We define a variable $z_S$ for every $S \subseteq V$ and let $\bo(S, T) = (\cl(S) \setminus S)\cap T$ denote the boundary of the set $S$ consisting of vertices in $T$  with all neighbors in $S$.   
Our main upper bound is the following.
\begin{theorem}\label{thm:lp_partial}
Consider the LP below.
\begin{align*}
\textrm{maximize}\quad & z_V \tag{Information theoretic LP for partial failure}\\
\textrm{s.t.}\quad & z_{\emptyset} = 0\\
& z_T - z_S \leq  |T\setminus S| -  (1-R_q(\delta))\cdot |\bo(S,T)| \quad \forall S \subseteq T \\
& z_S + z_T \geq z_{S \cap T} + z_{S \cup T} \quad \forall S,T
\end{align*}
The optimal solution to the above LP is an upper bound on $\scap_q(G,\delta)$. 
\end{theorem}
\begin{proof}
The proof follows the  same reasoning as the proof of the bound via information theoretic LP of Sec.~\ref{sec:infoLP}.
Indeed, the variable $z_S$ for every $S \subseteq V$ denote the entropy of $S$, $H(X_S)$. The last constraint on the LP above follows from sub-modularity. To establish the second constraint, first note that $|T\setminus \cl(S)| = |T\setminus S| -|\bo(S,T)|$. Now we have to show that for any $S\subseteq T$,
$$
 H(X_T) - H(X_S) \leq  |T\setminus \cl(S)| +  R_q(\delta)\cdot |\bo(S,T)|.
$$
To see this,
note that, 
\begin{align*}
H(X_T) &= H(X_S, X_{\bo(S,T)}, X_T) \\
&= H(X_S)+ H(X_{\bo(S,T)}|X_S) + H(X_T|X_{\bo(S,T)}, X_S)\\
& \le H(X_S) +  \sum_{i \in \bo(S,T)}H(X_{i}|X_S) + |T\setminus\cl(S)|\\
\Rightarrow H(X_T) - H(X_S) &\le |T\setminus\cl(S)| + \sum_{i \in \bo(S,T)}H(X_{i}|X_{N(i)})\\
& \le  |T\setminus\cl(S)| +  | \bo(S,T)| \lim_{m \to \infty} \frac1m \log_q A_q(m,\delta m+1)\\
& = |T\setminus\cl(S)| +  | \bo(S,T)| \cdot R_q(\delta),
\end{align*}
where the last two lines follow from Thm.~\ref{thm:partial} and the definition of $R_q(\delta)$.
\end{proof}


\paragraph{Odd Cycle Example.}
Consider an odd cycle with $n$ vertices where $n$ is odd. Below we show an example to illustrate the above bound on partial recovery capacity.

Consider the following subset of constraints: 
\begin{align*}
2 & \geq z_{\{1,3\}} - z_{\emptyset}\\
2 & \geq z_{\{2,4\}} - z_{\emptyset}\\
1 & \geq z_{\{i\}} \quad \forall i \in \{ 5, 6,\dots, n \}\\
R_q(\delta) & \geq z_{\{1,2,3\}} - z_{\{1,3\}}\\
R_q(\delta) & \geq z_{\{2,3,4\}} - z_{\{2,4\}}\\
z_{\{1,2,3\}} + z_{\{2,3,4\}} & \geq z_{\{2,3\}} + z_{\{1,2,3,4\}}\\
z_{\{2,3\}} + z_{\{5\}} + z_{\{7\}} + \dots + z_{\{n\}} & \geq z_{\{2,3,5,7,\ldots, n\}} + \frac{n-3}{2} z_{\emptyset} & \textrm{(a)}\\
z_{\{1,2,3,4\}} + z_{\{6\}} + z_{\{8\}} + \dots z_{\{n-1\}} & \geq z_{\{1,2,3,4,6,8,\ldots,n-1\}} + \frac{n-5}{2} z_{\emptyset} & \textrm{(b)}\\
(n-\frac{n+1}2) -(1-R_q(\delta))\frac{n-1}2 & \geq z_V - z_{\{2,3,5,7,\ldots,n\}}\\
(n-\frac{n+3}2) -(1-R_q(\delta))\frac{n-3}2 & \geq z_V - z_{\{1,2,3,4,6,8,\ldots,n-1\}}
\end{align*}
Equations (a) and (b) above are repeated applications of the inequality: $z_{S}+z_{T} \geq z_{S \cup T} + z_{\emptyset}$ if $S \cap T = \emptyset$. By summing up those constraints we get 
$$
n+ 2R_q(\delta) + R_q(\delta) (n-2) \geq 2z_V - 2z_{\emptyset}
$$ 
and thus 
$$
\scap_q(G,\delta) \leq z_V \leq \frac{n}{2}(1+R_q(\delta)),
$$
whenever $G$ is an odd cycle.

\subsection{Achievability bound}
A naive achievability bound on  $\scap_q(G,\delta)$ is given by,
$$
\scap_q(G,\delta) \ge n(1 - h_q(\delta)), \quad \delta \le 1/2, 
$$
where $h_q(x) \equiv x\log_q(q-1) - x\log_q x -(1-x) \log_q(1-x)$.
This amount of storage can be achieved by just using an error-correcting code of length $m$, distance $\delta m +1$, and rate $1-h_q(\delta)$
in each of the vertices. Such codes exist, by the Gilbert-Varshamov bound.
Also, 
$$
\scap_q(G,\delta) \ge 0, 
$$
for $0 \le \delta \le 1-1/q$.

This simple bound can be improved by more carefully designing a code. Our main result of this section is the following.
\begin{theorem}
Given a graph $G$, let  $\calC$ be the set of all cliques of $G$. The generalized clique packing  number $\fcc_\delta(G)$ is defined to be 
 the optimum  of the following linear program.  For $0 \le x_C\le 1, \forall C \in \calC$, 
\begin{align*}
 \max \sum_{C \in \calC} x_C(|C| -h_q(\delta)), \tag{$\fcc_\delta(G)$}
\end{align*}
such that,
$$
\sum_{C\in \calC: u \in C} x_C \le 1.
$$
Then,
$$
\scap_q(G,\delta) \ge \fcc_\delta(G), \quad \delta \le 1-1/q,
$$
and,
$$
\scap_q(G,\delta) \ge \fcc(G), \quad \delta > 1-1/q.
$$
\end{theorem}

\begin{proof}
First of all, notice that, $\scap_q(G,\delta) \ge \scap_q(G,1) =  \scap^{(q)}(G) \ge \fcc(G)$ where the last inequality follow from Lemma \ref{lem:scaplb}. Below therefore we only concentrate on the case when $\delta \le 1-\frac1q.$
We illustrate the proof of this theorem by constructing a sequence of error-correcting codes that serves our purpose.

First, we show that for any positive integer $d$ and a large enough positive integer $m$, there exists a linear error-correcting code of length $dm$ and dimension $dm - m h_q(\delta)$  that can correct any $\delta m$ erasures between coordinates $im+1$ and $(i+1)m$ for any $i\in \{0,1, \dots, d-1\}$.

Randomly and uniformly choose a $q$-ary parity check matrix of size $(dm - k) \times dm$ (that is, each coordinate of the 
matrix is chosen from $\{0,1,\dots, q-1\}$ with uniform probability). The probability that a vector of weight 
$\delta m$ is a codeword is $q^{-(dm-k)}$.
Now the probability that there exists such a codeword that is an uncorrectable erasure pattern of the above type is
$$
\le d \binom{m}{\delta m}(q-1)^{\delta m} q^{-(dm -k)} \le  \frac{d}{\sqrt{m}}q^{-(dm-k - m h_q(\delta))}, \delta \le 1-\frac1q.
$$  
Hence there exists such a code with dimension  $dm - m h_q(\delta)$ for any $\delta \le 1-\frac1q$ for large enough $m$.

Suppose, $C \in \calC$ be a clique of size $d = |C|$ in $G$. Use the linear error-correcting code of length $dm$ constructed above and store each block of $m$ coordinates in one of the vertices of the clique. If up to $\delta$ proportion of the content of any vertex is erased, it can be  recovered by accessing the other vertices. The total information stored in this clique is $dm - m h_q(\delta) = m (|C|- h_q(\delta))$.

Now let us find a partition of the graph into a collection of cliques $\{C_1, C_2, \dots, C_t\}$, such that each vertex belongs to at most one clique from the collection. For each clique $C_i, i =1, \dots, t,$, use an error-correcting code of length $|C_i|\cdot m$ to store $m (|C_i|- h_q(\delta))$ $q$-ary information. In this way the maximum amount of information that can be stored in the graph with the partial recovery condition is, $$
 \max \sum_{C \in \calC} x_C(|C| -h_q(\delta)),
$$
where $x_C \in \{0,1\}$ denotes whether the clique $C$ is in the collection, under the constraint that each vertex is in only one clique, i.e., 
$
\sum_{C\in \calC: u \in C} x_C \le 1.$ This is an integer linear program and 
$\scap_q(G,\delta)$ is at least the optimum value of this integer linear program for $\delta \le 1-\frac1q.$

Now, following an argument similar to  Lemma \ref{lem:scaplb}, we show that  the  integer linear program can be relaxed to a linear program and an achievability scheme still exists. For this, assume, $\{x_C\}_{C\in \calC}$ achieve $\fcc_\delta(G)$. Let $m$ be  such that ${x_C}\cdot m $ is integral and large enough for every $C$. For each clique $C=\{u_1,\ldots, u_{|C|}\}$ in the graph, in each vertex $u \in C$, store $m x_C$ $q$-ary symbols, such that the $|C|\cdot m x_C$-length vector in the clique $C$ is a codeword of an error correcting code of length $|C|\cdot m x_C$ that can carry  $m x_C (|C|- h_q(\delta))$ symbols of information, as noted above. Once we do this for all cliques, the number of $q$-ary symbols stored in vertex $u$ is $\sum_{C\in \calC: u \in C} x_C m \le m.$ The total amount of information stored in the graph is $\sum_{C \in \calC} m x_C (|C|- h_q(\delta)) = m \fcc_\delta(G)$.

This proves the claim.
\end{proof}

\vspace{0.2in}

\paragraph{Odd Cycle Example.}
Let us consider the example of $n$-cycle again where $n$ is an odd number.
Since the size of a fractional matching is $\frac{n}2$, we have
$$
\scap_q(G,\delta)  \ge \frac{n}{2}(2-h_q(\delta)), \quad \delta \le 1-\frac1q,
$$
and $\scap_q(G,\delta)  \ge \frac{n}{2}$ when $ \delta > 1-\frac1q$. Compare this with the impossibility bound that we have,
$$
\scap_q(G,\delta)  \le \frac{n}{2}(1+R_q(\delta)).
$$
It is widely conjectured that the optimal rate of an error-correcting code is given by
$$
R_q(\delta) = 1- h_q(\delta),
$$
for small $q$, which is also known as the Gilbert-Varshamov conjecture. If this conjecture is true, then our upper and lower bounds match exactly.
In particular, for large $q$ (i.e., $q \to \infty$), we have $h_q(\delta) \to \delta$ and $R_q(\delta) \to 1-\delta.$ Hence, our bounds match definitively 
 in the regime of large $q$.

\section{Conclusions}
\label{sec:conclusion}

Storage capacity is a natural problem of network coding and intimately related to the index coding problem which encapsulates the computational challenges of general network coding. In this paper we have viewed storage capacity as a natural information theoretic analog of vertex cover of graph. For some family of graphs, vertex cover is easier to approximate, such as the planar graphs; we see that, the storage capacity is also easier to estimate for these families. The relation to index coding also leads to approximation guarantees for index coding rate. We further illustrated an approach to bound the storage capacity of graphs in terms of a small number of vertex covers, which leads us to exactly quantify the storage capacity of cycles with chords.

In the last part of this paper, we provided one possible generalization of the storage capacity to partial recovery capacity. It is important to note that there are several other possible generalizations to this quantity that may be useful in practice. For example, one might consider recovery from failure of multiple vertices together from their combined neighborhood (analogous to the cooperative repair problem in distributed storage \cite{rawat2015cooperative}). In yet another scenario of  recovery from multiple  failures, a vertex failure may be recoverable from its neighborhood as long as at most $t\ge 0$ of its neighbors have also failed. We leave these generalizations as interesting future studies of storage capacity.

{ \small
\bibliographystyle{abbrv} \bibliography{storage}
}

\end{document}